\documentclass[12pt, conference, onecolumn]{IEEEtran}

\usepackage[utf8]{inputenc} 
\usepackage[T1]{fontenc}
\usepackage{url}
\usepackage{ifthen}
\usepackage{cite}
\usepackage[cmex10]{amsmath} 

\usepackage{amsmath}
\usepackage{amsthm}
\usepackage{amssymb}
\usepackage{multirow}
\usepackage{amsfonts}
\usepackage{graphicx}
\usepackage{algorithm}
\usepackage{url}
\usepackage{lipsum}

\usepackage{latexsym}
\usepackage{color}
\usepackage{graphicx}
\usepackage{threeparttable}
\usepackage{float}

\usepackage{algpseudocode}
\usepackage{etoolbox}
\usepackage{tikz}
\usetikzlibrary{tikzmark}
\usetikzlibrary{calc}
\usepackage{qtree}
\usepackage{tikz-qtree}

\errorcontextlines\maxdimen

\newcommand{\ALGtikzmarkcolor}{black}
\newcommand{\ALGtikzmarkextraindent}{4pt}
\newcommand{\ALGtikzmarkverticaloffsetstart}{-.5ex}
\newcommand{\ALGtikzmarkverticaloffsetend}{-.5ex}
\makeatletter
\newcounter{ALG@tikzmark@tempcnta}

\newcommand\ALG@tikzmark@start{%
    \global\let\ALG@tikzmark@last\ALG@tikzmark@starttext%
    \expandafter\edef\csname ALG@tikzmark@\theALG@nested\endcsname{\theALG@tikzmark@tempcnta}%
    \tikzmark{ALG@tikzmark@start@\csname ALG@tikzmark@\theALG@nested\endcsname}%
    \addtocounter{ALG@tikzmark@tempcnta}{1}%
}

\def\ALG@tikzmark@starttext{start}
\newcommand\ALG@tikzmark@end{%
    \ifx\ALG@tikzmark@last\ALG@tikzmark@starttext
    \else
        \tikzmark{ALG@tikzmark@end@\csname ALG@tikzmark@\theALG@nested\endcsname}%
        \tikz[overlay,remember picture] \draw[\ALGtikzmarkcolor] let \p{S}=($(pic cs:ALG@tikzmark@start@\csname ALG@tikzmark@\theALG@nested\endcsname)+(\ALGtikzmarkextraindent,\ALGtikzmarkverticaloffsetstart)$), \p{E}=($(pic cs:ALG@tikzmark@end@\csname ALG@tikzmark@\theALG@nested\endcsname)+(\ALGtikzmarkextraindent,\ALGtikzmarkverticaloffsetend)$) in (\x{S},\y{S})--(\x{S},\y{E});%
    \fi
    \gdef\ALG@tikzmark@last{end}%
}

\apptocmd{\ALG@beginblock}{\ALG@tikzmark@start}{}{\errmessage{failed to patch}}
\pretocmd{\ALG@endblock}{\ALG@tikzmark@end}{}{\errmessage{failed to patch}}
\algblockdefx[Foreach]{Foreach}{EndForeach}[1]{\textbf{foreach} #1 \textbf{do}}{\textbf{end foreach}}

\DeclareGraphicsExtensions{.tif,.pdf,.jpeg,.png,.jpg,.bmp,.svg,.eps,.EMF}

\newtheorem{theorem}{Theorem}
\newtheorem{definition}{Definition}

\newcommand{\cM}{\mathcal{M}}
\newcommand{\cA}{\mathcal{A}}

\newcommand{\cS}{\mathcal{S}}
\newcommand{\cT}{\mathcal{T}}

\newcommand{\cG}{\mathcal{G}}

\newcommand{\bX}{\mathbf{x}}
\newcommand{\bY}{\mathbf{y}}

\newcommand{\bV}{\mathbf{v}}
\newcommand{\bU}{\mathbf{u}}

\newcommand{\poly}{\mathsf{poly}}

\newcommand{\diag}{\mathsf{diag}}

\newcommand{\supp}{\mathsf{supp}}
\newcommand{\test}{\mathsf{test}}
\newcommand{\dec}{\mathsf{dec}}

\newcommand{\sfA}{\mathsf{A}}

\newcommand{\rTGT}{\mathrm{TGT}}
\newcommand{\rGTTI}{\mathrm{GTTI}}

\begin{document}

\title{A framework for generalized group testing with inhibitors and its potential application in neuroscience\\[.2ex] 
  {\normalfont\large 
	Thach V. Bui\IEEEauthorrefmark{1}, Minoru Kuribayashi\IEEEauthorrefmark{3}, Mahdi Cheraghchi\IEEEauthorrefmark{4}, and Isao Echizen\IEEEauthorrefmark{1}\IEEEauthorrefmark{2}}\\[-1.5ex]}

\author{\IEEEauthorblockA{\IEEEauthorrefmark{1}SOKENDAI (The \\Graduate University \\for Advanced \\Studies), Hayama, \\Kanagawa, Japan\\ bvthach@nii.ac.jp}
\and
\IEEEauthorblockA{\IEEEauthorrefmark{3}Graduate School\\ of Natural Science\\ and Technology, \\Okayama University, \\Okayama, Japan\\kminoru@okayama-u.ac.jp}
\and
\IEEEauthorblockA{\IEEEauthorrefmark{4}Department of \\Computing, Imperial \\College London, UK\\m.cheraghchi@imperial.ac.uk}
\and
\IEEEauthorblockA{\IEEEauthorrefmark{2}National Institute\\ of Informatics, \\Tokyo, Japan \\ iechizen@nii.ac.jp}}

\maketitle

\thispagestyle{plain}
\pagestyle{plain}

\begin{abstract}
The main goal of group testing with inhibitors (GTI) is to efficiently identify a small number of defective items and inhibitor items in a large set of items. A test on a subset of items is positive if the subset satisfies some specific properties. Inhibitor items cancel the effects of defective items, which often make the outcome of a test containing defective items negative. Different GTI models can be formulated by considering how specific properties have different cancellation effects. This work introduces generalized GTI (GGTI) in which a new type of items is added, i.e., hybrid items. A hybrid item plays the roles of both defectives items and inhibitor items. Since the number of instances of GGTI is large (more than 7 million), we introduce a framework for classifying all types of items non-adaptively, i.e., all tests are designed in advance. We then explain how GGTI can be used to classify neurons in neuroscience. Finally, we show how to realize our proposed scheme in practice.
\end{abstract}

\section{Introduction}
\label{sec:intro}

\subsection{Group Testing with inhibitors}
\label{sub:intro:GT}

\textbf{GROUP TESTING:} Identifying a small group of items satisfying \textit{specific properties} in a colossal group of $n$ items is the main problem in group testing. Such items are usually referred to as \textit{defective items}, and the other items are usually referred to as \textit{negative items}. The classification of group testing depends on the classification of defective items. Suppose that there are $n$ items indexed from $1$ to $n$ and that the defective set $D \subset [n] = \{1, \ldots, n \}$. A test on a subset of $n$ items is designed to determine whether the subset satisfies the properties. If the properties are satisfied, the test outcome is \textit{positive}. Otherwise, the test outcome is negative. In general, how $D$ is determined to be a defective set and how members of the defective set $D$ present in a test defines the test outcome. Note that the definition of $D$ is inseparable with the specific properties. Here we assume that the upper bound of the cardinality of $D$ is known, i.e., $|D| \leq d \ll n.$

There are two main approaches to test design: adaptive and non-adaptive. In adaptive design, there are several stages of testing and a test depends on the outcomes of previous tests. This approach usually achieves the minimum number of tests needed; the defective items are finally revealed by the last test. However, since there are several stages, it generally consumes lots of time. The non-adaptive design approach reduces the testing time because all tests are independent and designed a priori. However, this approach generally requires a larger number of tests. 

Another important concern that should be considered is noise in the test outcome. In noisy setting, a test outcome can flip from positive to negative and vice versa. The number of tests required in a noisy setting is usually larger than that in a noiseless setting.

The procedure for obtaining outcomes by testing subsets of items is called \textit{encoding} and that for classifying the items from the test outcomes is called \textit{decoding}. It is desirable to minimize the number of tests for the encoding procedure and minimize the time required for classifying the items. Other criteria are also considered for specific problems. The paradigm of the theoretical model of group testing is illustrated in Fig.~\ref{fig:GGT}. The practical model will be addressed in latter sections.

\begin{figure}[t]
\centering
\includegraphics[scale=0.545]{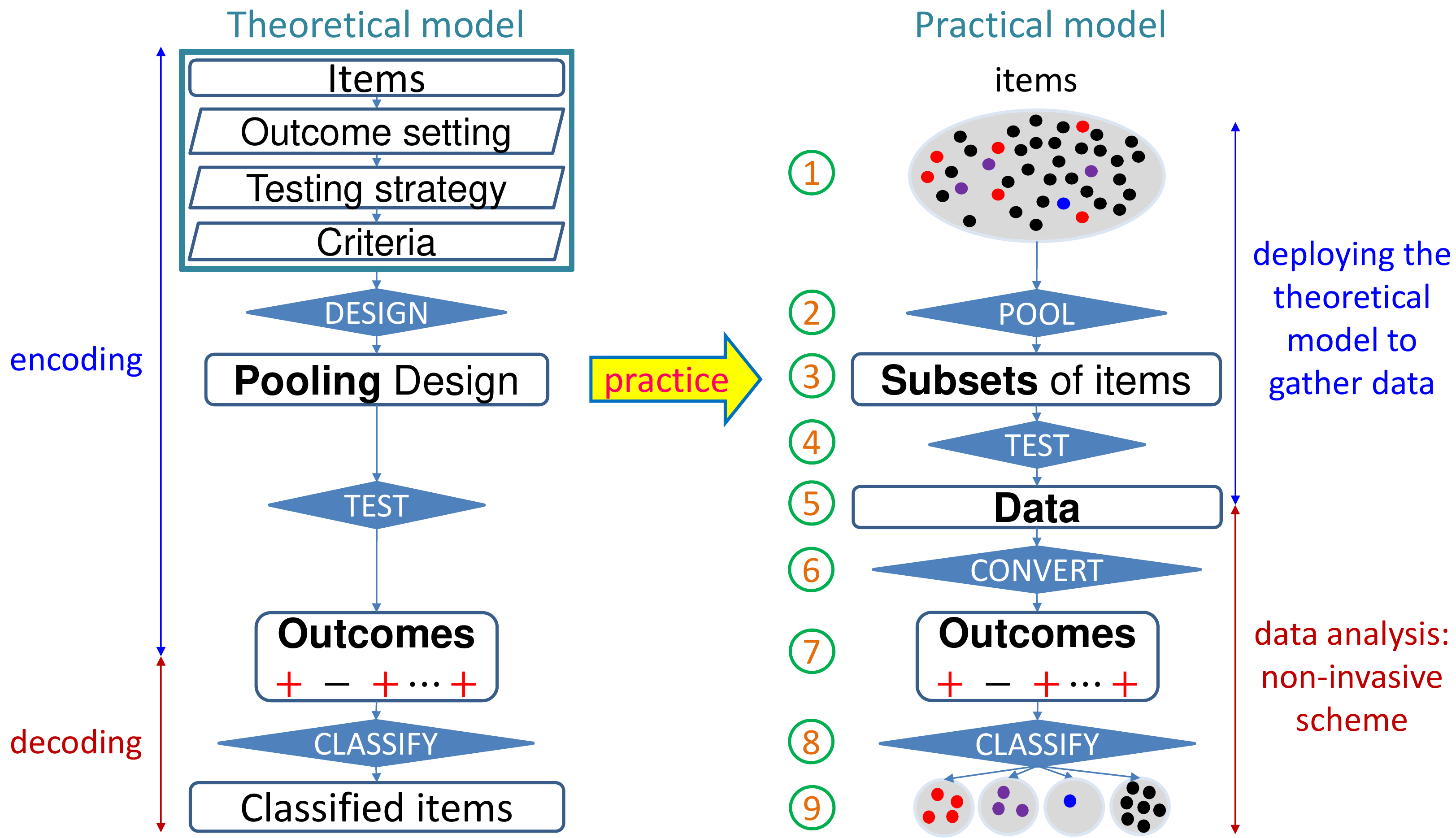}
\caption{Generalized group testing with inhibitors paradigm encompasses two procedures: encoding and decoding. For the theoretical model, the objective of encoding procedure is to create a pooling design (i.e., a measurement matrix) and then do tests on it to get outcomes. From the outcomes, the items are classified via the decoding procedure. There are nine steps in the practical model while steps 4 to 6 merge into one step in the theoretical model. In the practical model, a set of items might consist of \textcolor{red}{defective} items ($\color{red} \bullet$), \textcolor{violet}{inhibitors} ($\color{violet} \bullet$), \textcolor{blue}{hybrid} items ($\color{blue} \bullet$), and negative items ($\bullet$). Data gathered from tests must be converted into binary outcomes before proceeding the decoding procedure to classify the items.}
\label{fig:GGT}
\end{figure}

Based on the types of defective items, their classification is as follows and will be addressed in the following sections.
\begin{align}
\begin{tabular}{@{}c@{}} Complex group testing \\ (Complex defectives) \end{tabular} \xrightarrow{\text{reduce}} \begin{tabular}{@{}c@{}} Threshold group testing \\ (Threshold defectives) \end{tabular} \xrightarrow{\text{reduce}} \begin{tabular}{@{}c@{}} Classical group testing \\ (Classical defectives) \end{tabular} \nonumber
\end{align}

\subsubsection{Classical group testing}
\label{subsub:CGT}

In classical group testing (CGT), the outcome of a test on a subset of items is positive if the subset has at least one item in $D$, and negative otherwise. CGT has been intensively studied since its inception~\cite{dorfman1943detection}. With adaptive design, the number of tests is $O(d \log{n})$~\cite{du2000combinatorial}. However, this design is problematic because it can take too much time to carry out the tests. With non-adaptive design, the number of tests is $O(d^2 \log{n})$~\cite{d1982bounds,du2000combinatorial}. Porat and Rothschild~\cite{porat2008explicit} first proposed explicit nonadaptive schemes to achieve this bound. However, there is no sublinear decoding algorithm associated with their schemes. For exact reconstruction in a noisy setting, Ngo et al.~\cite{ngo2011efficiently} proposed a scheme for identifying up to $d$ defective items in time $\poly(d, \log{n}, z)$ with $O(d^2 \log{n} + dz)$ tests in the presence of up to $z$ erroneous outcomes. If false positives are allowed in the defective set, Cheraghchi~\cite{cheraghchi2013noise} proved that the number of tests can be as low as $O(d^{1 + o(1)} \log{n})$ with a decoding time of $\poly(d, \log{n})$. Following Cheraghchi's idea that the number of tests can be reduced if exact construction is not required, Cai et al.~\cite{cai2013grotesque} used probabilistic schemes that need $O(d \log{d} \cdot \log{n})$ tests to find the defective items in time $O(d(\log{n} + \log^2{d}))$ w.h.p. Using numerical results, Bui et al.~\cite{bui2018efficient} recently showed that $O \left(\frac{d^2 \log^2{n}}{W^2(d\log{n})} \right)$ is the minimum number of tests, unlike the previously reported~\cite{{ngo2011efficiently}}, where $W(x) = \Theta \left( \log{x} - \log{\log{x}} \right)$. The decoding time of this scheme is also low: $O(d^6 \log^6{n})$. CGT is widely used in various fields such as DNA library screening~\cite{ngo2000survey}, compressed sensing~\cite{Atia2012:Boolean}, graph constraining~\cite{cheraghchi2012graph}, and similarity searching~\cite{iscen2016efficient}.

\subsubsection{Threshold group testing}
\label{subsub:TGT}
In threshold group testing (TGT), given two integer parameters $0 \leq \ell < u \leq d$, the outcome of a test on a subset of items is negative if the subset has up to $\ell$ items in $D$, is positive if the subset has at least $u$ items in $D$, and arbitrary otherwise.

The two parameters $\ell$ and $u$ are called the lower threshold and upper threshold, respectively. This TGT is thus denoted as $\rTGT(\ell, u)$. Let $g = u - \ell - 1$ be the gap in $\rTGT(\ell, u)$. TGT has no gap when $g = 0$. When $u = 1$, TGT reduces to CGT. To avoid confusion with the definition of defective items given in section~\ref{subsub:CGT}, we call defective items in TGT \textit{threshold defective items.}

Damaschke~\cite{damaschke2006threshold} introduced TGT in 2006. By using $\binom{n}{u}$ non-adaptive tests, he showed that a set of positive items could have up to $g$ false positives and $g$ false negatives. The number of tests was then decreased to $O(d^{g + 2} \log{d} \cdot \log{\frac{n}{d}})$ by Cheraghchi~\cite{cheraghchi2013improved}. When the number of defective items is known, e.g., $d$, the number of tests can be reduced to $O(d^{1.5} \log \frac{n}{d})$ for $g = 0$~\cite{de2017subquadratic} or $O(g^2 d \log{n}) + O(d \log{\frac{1}{\epsilon}})$ for $\epsilon > 0$~\cite{chan2013near}. Although the number of tests is low, this approach is rarely applied because of the condition on the number of defectives items.

Most work on TGT has focus on the number of tests; there has been relatively little on the decoding procedure. Chen and Fu~\cite{chen2009nonadaptive} use $O \left( z \cdot \frac{(d + u - \ell)^{d + 1}}{u^u (d - \ell)^{d - \ell})} \log{\frac{n}{d + u - \ell}} \right)$ tests to find defective items in time $O(n^u \log{n})$ when there were at most $z$ erroneous outcomes. Since the number of tests and the decoding time become huge as $n$ increases, this approach is mostly impractical. Chan et al.~\cite{chan2013stochastic} presented a randomized algorithm for finding defective items with $O\left( \log{\frac{1}{\epsilon}} \cdot d\sqrt{u} \log{n}\right)$ tests in time $O(g^2 n\log{n} + n \log{\frac{1}{\epsilon}})$ given that the number of defective items is exactly $d$ and $u = o(d)$. Again, these conditions are too strict for practical applications, and the cost of decoding increases with $n$. Bui et al.~\cite{bui2018efficiently} recently proposed a scheme for finding up to $d$ defective items using $t = O \left(  \left( \frac{d}{u} \right)^u \left( \frac{d}{d -u} \right)^{d - u} d^3 \log{n} \cdot \log{\frac{n}{d}} \right)$ tests in sub-linear time $O( t \times \poly(d, \log{n}))$ when $g = 0$. The authors\cite{bui2019improved} later improved the number of tests to $t = O\left( \frac{d^4 \log^3{n}}{W^2(d\log{n})} \right)$ and the decoding time to $O( t \times \poly(d, \log{n})).$ Using stochastic approach, by setting $d = O(n^\beta)$ for $\beta \in (0, 1)$ and $u = o(d)$, Reisizadeh et al. \cite{reisizadeh2018sub} used $\Theta(\sqrt{u} d \ln^3{n} )$ tests to identify all defective items in time $O(u^{1.5} d \ln^4{n} )$ w.h.p with the aid of a $O(u \ln{n}) \times \binom{n}{u}$ look-up matrix. These settings are mostly unfeasible when $n$ or $u$ grows.

\subsubsection{Complex group testing}
\label{subsub:complexGT}

In complex group testing (CmplxGT), $D$ is decomposed into $c$ smaller subsets $D_1, \ldots, D_c$, such that:
\begin{itemize}
\item $D = D_1 \cup \ldots \cup D_c$.
\item Any $D_a$ is the defective set in $\rTGT(\ell_a, u_a)$, where $0 \leq \ell_a < u_a \leq u \leq d$ for $a = 1, \ldots, c$.
\end{itemize}

A test outcome on a subset of items is positive if the subset contains items in $D_a$ and the outcome of a test on these items is positive, for some $a \in [c]$. Otherwise, the test outcome is negative.

When every defective subset of $D$ has no gap, CmplxGT is called \textit{complex group testing without gap}. When $c = 1$, CmplxGT reduces to TGT. When $c = 1$ and $u_1 = 1$, CmplxGT reduces to CGT. To avoid confusion with the definitions of defective items given in sections~\ref{subsub:CGT} and~\ref{subsub:TGT}, we call defective items in CmplxGT \textit{complex defective items.}

CmplxGT orginated in molecular biology~\cite{torney1999sets}. In this setting, $\ell_a = u_a - 1 = |D_a| - 1$. Chen et al.~\cite{chen2008upper} restated this problem as complex group testing. There has been some work on CmplxGT~\cite{chang2010identification,chen2008upper,chin2013non}. Chen et al.~\cite{chen2008upper} showed that the number of non-adaptive tests is $O \left( z \left( \frac{d + u}{u} \right)^u \left( \frac{d + u}{d} \right)^d (d + u) \log{\frac{n}{d + u}} \right)$, where $z$ is the maximum number of erroneous outcomes. Without considering errors in the test outcomes, Chin et al.~\cite{chin2013non} improved this bound to $O \left( (c + d)^{c + d + 1} \log{n} /(c^c d^d) \right)$. These bounds increase as $u$ or $d$ increases.


\vspace{2mm}
\noindent
\textbf{INHIBITORS:} Recent advances in the definition of group testing have added a new type of item: inhibitors. The manner in which the outcome of a test is positive defines the type of defective item, and the manner in which the outcome is negative defines the type of inhibitor item. An item is considered to be an inhibitor if it interferes with the identification of the defective items. The inhibitor set is denoted as $H$, where $|H| \leq h \ll n$. Similar to the classification of defective items, there are three types of inhibitors: classical (dictator) inhibitors, threshold inhibitors, and complex inhibitors. If a subset does not contain any item in $H$, the test outcome on the subset depends on the properties of defective items presented in the subset. Therefore, we only describe what the outcome of a test on a subset of items is when the subset has at least one item in $H$. In this model, we have three sets out of $n$ items: defective set $D$ ($|D| \leq d \ll n$), inhibitor set $H$, and negative set $[n] \setminus D \cup H$. There are three classifications of inhibitor items:

\begin{align}
\text{Complex inhibitors} \xrightarrow{\text{reduce}} \text{Threshold inhibitors} \xrightarrow{\text{reduce}} \text{Classical inhibitors} \nonumber
\end{align}

For testing design, researchers often use a non-adaptive design to reduce testing time in group testing with inhibitors (GTI). For decoding, there are also two approaches: 1) identify defective items only and 2) identify both defective items and inhibitor items.

\subsubsection{Classical (dictator) inhibitors}
\label{subsub:inhi:classical}

In group testing with classical inhibitors (GTDI), if a subset of items contains at least one inhibitor item in $H$ then the test outcome on that subset is negative~\cite{chang2010identification,de1998improved,farach1997group,ganesan2015non,hwang2003error}. Let $z$ be the maximum number of erroneous outcomes and $\lambda = (d + h) \log{n}/W((d + h)\log{n}) + z$, where $W(x) = \Theta \left( \log{x} - \log{\log{x}} \right)$. The following details follow the setting that $D$ is a defective set in CGT and $H$ is an inhibitor set in GTDI. More specifically, the outcome of a test on a subset is positive if the subset has no inhibitor items and has at least one defective items, and negative otherwise. To identify defective items only, Chang et al.~\cite{chang2010identification} proposed a scheme using $O((d + h + z)^2\log{n})$ tests in time $O((d + h + z)^2 n\log{n})$. For identifying both defective items and inhibitors, they proposed a scheme using $O( z(d + h)^3\log{n})$ tests in time $O( e(d+h)^3 n \log_2{n})$. Without considering erroneous outcomes, Ganesan et al.~\cite{ganesan2015non} use $O((d + h) \log{n})$ tests to identify defective items in time $O((d + h)n \log{n})$ by using a probabilistic scheme. It took $O((d + h^2) \log{n})$ tests to identify both defective and inhibitor items in time $O((d + h^2)n \log{n})$. Bui el at.~\cite{bui2018sublinear} recently proposed a scheme that uses $\Theta \left( \lambda^2 \log{n} \right)$ tests to identify only the defective set $D$ in time $ O \left( \lambda^5 \log{n}/(d+h)^2 \right)$. They also proposed another scheme for identifying both defective and inhibitor items in time $O \left( d \lambda^6 \times \max \left\{ \lambda/(d+h)^2, 1 \right\} \right)$ using $O \left( \lambda^3 \log{n} \right)$ tests.

\subsubsection{Threshold inhibitors}
\label{subsub:inhi:thresh}

In group testing with threshold inhibitors (GTTI), the outcome of a test on a subset which has at least one item in $H$ satisfies the following conditions.

\begin{itemize}
\item If the number of inhibitor items in a subset is at least $ui$, the outcome of a test on the subset is negative.
\item If the number of inhibitor items in a subset is up to $li$, the outcome of a test on the subset depends only on the type and number of defective items.
\item If the number of inhibitor items in a subset is more than $li$ and less than $ui$, the outcome of a test on the subset depends on the type of defective items, the number of defective items, and the properties of the inhibitor items in the subset.
\end{itemize}

This threshold inhibitor model is denoted as $\rGTTI(li, ui)$. The two parameters, $li$ and $ui$, are the lower threshold and upper threshold. A gap is denoted as $gi = ui - li - 1$. When $ui = 1$, GTTI reduces to GTDI. To avoid confusion with the definition of inhibitor items given in section~\ref{subsub:inhi:classical}, we call inhibitor items in GTTI \textit{threshold inhibitor items.}

Two previous works~\cite{chang2010identification,chang2013threshold} specified GTTI with a formal definition of defective items. Both assumed that there is no gap in $\rGTTI(li, ui)$, i.e., $gi = 0$, and that there are up to $e$ erroneous outcomes. The first considered two models of defective items: CGT and CmplxGT without a gap. When $D$ is a defective set in CGT, there is a non-adaptive algorithm that can be used to classify all $n$ items. When $D$ is the defective set in CmplxGT, the authors can only recover complex defective items. The second work considered $D$ to be the defective set in TGT. Then all defective items can be identified in time $O \left( e n^d (d + h)^{u + 1} \log (n/(d + h)) \right)$.

\subsubsection{Complex inhibitors}
\label{subsub:inhi:cmplx} 

In group testing with complex inhibitors (GTCI), $H$ is decomposed into $ci$ smaller subsets $H_1, \ldots, H_{ci}$, such that:
\begin{itemize}
\item $H = H_1 \cup \ldots \cup H_{ci}$.
\item Any $H_a$ is the inhibitor set in $\rGTTI(li_a, ui_a)$, where $0 \leq li_a < ui_a \leq h$ for $a = 1, \ldots, ci$.
\end{itemize}

If a subset has at least $ui_a$ inhibitor items in $H_a$ for some $a \in [1, ci]$, then the outcome of a test on the subset is negative. Otherwise, the outcome of a test on a subset depends on the types and the number of inhibitors and defective items in the subset.

When $ci = 1$, GTCI reduces to GTTI. When $ci = 1$ and $ui_1 = 1$, GTCI reduces to GTDI. To avoid confusion with the inhibitor items defined in sections~\ref{subsub:inhi:classical} and~\ref{subsub:inhi:thresh}, we call defective items in GTCI \textit{complex inhibitor items.}

There has been only one work dealing with complex inhibitors up to date~\cite{zhao2016group}. The authors considered two error-tolerant models: CGT with complex inhibitors and TGT with complex inhibitors. Specifically, each small subset $H_a$ is the inhibitor set in $\rGTTI(ui_a - 1, ui_a)$ for $a = 1, \ldots, ci$. Their ultimate goal was to identify the defective $D$ efficiently while minimizing the number of tests. They could not identify both defective items and inhibitors items. Let $mi = \sum_{a = 1}^{ci} ui_a$, $ki = mi - ci + 1$, and let $z$ be the maximum erroneous outcomes. The defective set $D$ can be recovered using $O \left(z \log{n} \cdot \exp(d, mi, ki, h) \right)$ tests in time $O(n^{mi + u} \log{n})$.

\subsubsection{Hybrid items}
\label{subsub:inhi:hybrid}

To generalize GTI, we introduce a new type of item: hybrid items. A hybrid item can be either defective or inhibitory. Under \textit{certain conditions}, it is defective (resp., inhibitory) because it satisfies the properties of a defective item (resp., an inhibitory item). The formal definition of ``certain conditions'' is left for future work. Since there are three types of defective items and three types of inhibitor items, hybrid items consists of nine types.

\subsection{Action potentials in neuroscience}
\label{sub:AP}

When you listen to a funny story, you often laugh. Your nervous system consisting of the central nervous system (CNS) and the peripheral nervous system (PNS) is responsible for this action. The CNS is located in the brain and spinal cord and is encased in bone. Neurons in the PNS travel through or lie on top of muscle, organ, and skin tissue. The primary purpose of the CNS is to organize and analyze signals from the sensory and motor neurons of the PNS, allowing us to observe and react to the environment. The central purpose of the PNS is to follow the commands of the CNS by changing motor output.

The nervous system is mainly regulated by three types of neurons: sensory neurons, interneurons, and motor neurons. Sensory neurons conduct signals from inside and outside the body such as those responsible for taste and vision to the CNS via stimulus receptors. Motor neurons convey signals from the CNS to the effector cells such as the muscles and glands. Finally, interneurons, which are distributed entirely within the CNS, interconnect sensory neurons and motor neurons. There are approximately 86 billion neurons in a brain~\cite{herculano2009human}.

Now, going back to the mechanism underlying your laugh. The voice from the speaker reaches the sensory organs in your ear. The signals induced in those organs propagate to the sensory neurons, which connect with those in your spinal cord. The sensory neurons then generate signals that propagate to the interneurons in your cerebral cortex. The brain cortex processes the signals received and then decides to move facial muscles via motor neurons. As a result, you laugh.

The signals in the laugh reaction chain are \textbf{action potentials} (APs) in the cells. There is an electrical potential difference between the inside of a cell and the surrounding environment. When a neuron is at rest, its membrane potential is about $-70$mV. When a neuron is active, an AP is caused by rapid depolarization of the membrane beyond threshold. The threshold is typically about $-55$mV. Therefore, an AP is an ``\textbf{all or none}'' phenomenon. This means that once the membrane has become depolarized and reaches the threshold, an AP will occur.

Next we describe how an AP is generated in terms of neuron interaction. A neuron contacts and communicates with other neurons by creating special sites called \textit{synapses} using its axon and dendrites. Information, usually in the form of chemical substances called \textit{neurotransmitters}, generally flows in one direction, from a source neuron to a target neuron. The source neuron is said to be \textit{presynaptic}, and the target neuron is said to be \textit{postsynaptic}. To generate an AP in a postsynaptic neuron, several presynaptic neurons of a target neuron release neurotransmitters. The type(s) of neurotransmitters released by a presynaptic neuron defines its neuron type in terms of neurotransmitters. Most excitations in the cortex are generated by neurons releasing glutamate. Most inhibitions are generated by neurons releasing GABA. A more detailed explanation is available elsewhere~\cite{bear2016neuroscience}.

Normally, there are two classifications of neurons based on the postsynaptic potential: excitatory postsynaptic potential (EPSP) and inhibitory postsynaptic potential (IPSP). A neuron is \textit{excitatory} (resp., \textit{inhibitory}) if it makes EPSPs (resp., IPSPs) at its postsynaptic neurons. Most neurons are unable to play both excitatory and inhibitory roles~\cite{dale1935pharmacology}. However, some neurons can be both excitatory and inhibitory~\cite{chavas2003coexistence,uchida2014bilingual}. Such neuron are said to be \textit{hybrid}. Moreover, some neurons may play neither role for a certain stimulus, so we have another type of neuron: negative neurons. We thus consider four types of neurons \textit{for a stimulus}: excitatory neurons, inhibitory neurons, hybrid neurons, and negative neurons. Note that synaptic excitation and inhibition are inseparable events, even for the simplest sensory stimulus like a brief tone~\cite{tan2004tone}. Any imbalance in synaptic excitation and inhibition may cause disorders such as schizophrenia and epilepsy. Another important point is that a postsynaptic neuron can receive neurotransmitters from more than one neuron, i.e., multiple synaptic potentials merge within one postsynaptic neuron. This process is called \textit{synaptic integration}.

The main challenges related to APs are clarifying the mechanisms underlying how an AP is generated in a neuron and clarifying the mechanism of the interaction between neurons. A preliminary step in meeting these challenges is determining the neuron type corresponding to a stimulus. We have developed a scheme for classifying neurons for the stimulus.

\subsection{Contributions}
\label{sub:intro:contri}

In this work, we make four contributions. First, we generalize group testing with inhibitors by introducing a new type of item: hybrid items that can play the role of a defective item or that of an inhibitor item. Second, we present an encoding/decoding framework for generalized group testing with inhibitors (GGTI). Third, we introduce a mapping between GGTI and neuron classification. Finally, we show how to carry out experiments for classifying neurons by using GGTI in practice.

\subsection{Applications}
\label{sub:apps}
In this section, we outline two scenarios in which our contributions could play an important role and describe how they improve on earlier work.

\subsubsection{Medicine} Stimuli can be generated both inside and outside the body. Our objective is to locate which neurons are responsible for sensing/responding to a particular stimulus and then determine the types of those neurons. There are numerous potential applications of this capability. We mention only a few of them here. When a patient undergoes major surgery, the patient's body is usually completely anesthetized so that the patient is unconscious before and during surgery. The main post-operative task is to restore the patient's consciousness. If we could localize the neurons responsible for reaction in the part of the patient's body that is the surgical target, we could anesthetize only those neurons. This might obviate the need to restore consciousness because it might be possible for the patient to remain conscious during surgery. This would lessen the risk posed by surgery. This is the most optimistic application in treatment by using classified neurons.

Another potential application is the treatment of such disorders as acrophobia, schizophrenia, and epilepsy. Identification of the neurons responsible for these disorders would enable application of personalized treatment. A final example is the creation of a brain stimulus paradigm and then identifying the mechanism underlying that stimulus.

\subsubsection{Neural networks} Spike neural networks~\cite{hodgkin1952quantitative,gerstner2002spiking} have been intensively studied to mimic how a brain works. They have been applied in a wide range of machine learning, including computer vision~\cite{krizhevsky2012imagenet} and pattern recognition~\cite{hopfield1995pattern}. In biology~\cite{markram2015reconstruction}, the usual way is to identify the types of neurons by \textit{exhaustively examining} all neurons then build the connectivity between the neurons. In theoretical computer science, Lynch et al.~\cite{lynch2017computational} presented a computational model for investigating how many inhibitory neurons needed to be added into a network consisting of input neurons and corresponding output neurons. Identifying the firing neuron was their ultimate goal. In fact, the authors also needed to examine each neuron to achieve their goal.

Let consider each examination on a neuron as a test. Then the number of tests is linear to the number of neurons in the network in previous work. By using group testing, we could possibly reduce such number of tests. The details of this solution would be addressed in Sections~\ref{sec:neuro} and~\ref{sec:design}.

\subsection{Techniques}
\label{sub:intro:tech}

Our key techniques involve two concomitant variates: the tensor product and the divide-and-conquer strategy. The techniques are applied to the encoding and decoding of GGTI. The generalization of group testing with inhibitors and its mapping to neuron classification are done in terms of abstract formalization. Let $\circledcirc$ be the tensor product notation. Given an $f \times n$ matrix $\cA = (a_{ij})$ and an $s \times n$ matrix $\cS = (s_{ij})$, the tensor product of $\cA$ and $\cS$ is defined as
\begin{align}
\mathcal{R} = \cA \circledcirc \cS := \begin{bmatrix}
\cS \times \mathrm{diag}(\cA_{1, *}) \\
\vdots \\
\cS \times \mathrm{diag}(\cA_{f, *})
\end{bmatrix} = \begin{bmatrix}
a_{11} \cS_1 & \ldots & a_{1n} \cS_n \\
\vdots & \ddots & \vdots \\
a_{f1} \cS_1 & \ldots & a_{fn} \cS_n
\end{bmatrix}, \label{eqn:tensor}
\end{align}
where $\mathrm{diag}(.)$ is the diagonal matrix constructed from the input vector, and $\cA_{h, *} = (a_{h1}, \ldots, a_{hn})$ is the $h$th row of $\cA$ for $h = 1, \ldots, f$. The size of $\mathcal{R}$ is $r \times n$, where $r = f \times s$.

\section{Preliminaries}
\label{sec:pre}

For consistency, we use capital calligraphic letters for binary matrices, non-capital letters for scalars, capital letters for sets, and bold letters for vectors. Here are some notations used:
\begin{enumerate}
\item $n, \bX = (x_1, \ldots, x_n)^T$: number of items and representation vector of $n$ items.
\item $\cT_j, \cT_{i, *}, \cG_{i, *}, \cM_{i,*}, \cM_j$: column $j$ of matrix $\cT$, row $i$ of matrix $\cT$, row $i$ of matrix $\cG$, row $i$ of matrix $\cM$, and column $j$ of matrix $\cM$, respectively.
\item $\supp(\bV) = \{j \mid v_j \neq 0 \}$: support set of vector $\bV = (v_1, \ldots, v_w)$.
\item $\diag(\cG_{i, *}) = \diag(g_{i1}, \ldots, g_{in})$: diagonal matrix constructed by input vector $\cG_{i, *}$.
\end{enumerate}


As we outlined in section~\ref{sec:intro}, there are two main approaches in testing design: adaptive and non-adaptive. In adaptive design, there are several stages of testing, and the outcome of a test depends on the outcomes of previous tests. Since adaptive design is time consuming, it is not preferred in practice. Therefore, we focus on non-adaptive design in which all tests are independent and designed a priori. The tests can be represented as follows. Let $n, D, H$, and $B$ be the number of items, the defective set, the inhibitor set, and the hybrid set, where $|D| \leq d$, $|H| \leq h$, $|B| \leq b$, and $1 \leq d + h + b \ll n$. Let $\bX = (x_1, \ldots, x_n)^T$ be the representative vector of $n$ items, where $x_j = 0$ means that item $j$ is negative and that item $j$ is either defective, inhibitory, or hybrid, otherwise. Note that we do not specify which values represent defective, inhibitor, and hybrid. For a $t \times n$ binary measurement matrix $\cT = (t_{ij})$, item $j$ is represented by column $\cT_j$, test $i$ is represented by row $i$, $t_{ij} = 1$ if and only if item $j$ belongs to test $i$, and $t_{ij} = 0$ otherwise.

Let $\test(S)$ be the test on subset $S \subseteq [n]$. The outcome of the test is either positive (1) or negative (0) and depends on the definition of $D$, $H$, $B$, and $S$. The tests on $n$ items using $\cT$ is defined as
\begin{equation}
\mathbf{y} = \cT \bullet \mathbf{x} = \begin{bmatrix}
\test \left( \supp(\cT_{1, *}) \cap \supp(\bX) \right) \\
\vdots \\
\test \left( \supp(\cT_{t, *}) \cap \supp(\bX) \right) \\
\end{bmatrix} = \begin{bmatrix}
y_1 \\
\vdots \\
y_t
\end{bmatrix}, \label{encMeasurement}
\end{equation}
where $y_i = \test \left( \supp(\cT_{i, *}) \cap \supp(\bX) \right)$ is the outcome of test $i$ (on subset $\supp(\cT_{i, *})$). The procedure to get $\bY$ is called \textit{the encoding procedure.} It includes the construction procedure, which is used to get measurement matrix $\cT$. The procedure to recover $\bX$ from $\bY$ and $\cT$ is called \textit{the decoding procedure.} Our objective is to design measurement matrix $\cT$ with the small number of tests such that $\bX$ can be recovered efficiently when outcome vector $\mathbf{y}$ is observed.


\section{Generalized group testing with inhibitors}
\label{sec:framework}

\subsection{Model}
\label{sub:frameModel}

In this section, we generalize GTI by introducing a new type of item: hybrid items. In GTI, defective items only tend to make the outcome of a test on them positive, and inhibitor items only tend to make the outcome of a test on them negative. Hybrid items, however, can make the outcome of a test on them either positive or negative. The theoretical model of GGTI paradigm is illustrated in Figure~\ref{fig:GGT}. It has encompassed two procedures: encoding and decoding. Encoding includes designing a measurement matrix by choosing the types of items, criteria, noise setting in the test outcomes, and testing strategy. Decoding is the classification of items on the basis of these choices.

As discussed in section~\ref{sec:intro}, there are three types of defective items, three types of inhibitor items, and at least one type of hybrid items. For the noise setting, there are two types: noisy and noiseless. Similarly, there are two types of testing design: adaptive and non-adaptive. Yet, six criteria are considered here: construction type, decoding time, number of tests, space to generate an entry, time to generate an entry, and zero-gap. Construction type is usually random, i.e., the matrix is generated randomly, or non-random, i.e., the matrix is generated nonrandomly. The ``zero-gap'' criterion is essentially helpful in biological applications. It would reduce time to pick items when making tests. Formally, a $t \times n$ matrix $\cT$ has zero-gaps if for any permutation of the rows in $\cT$, $\sum_{i = 1}^{t - 1} \mathrm{wt}(\cT_{i, *}, \cT_{i + 1, *})$ is minimum, where $\mathrm{wt}(\bU, \bV)$ is the Hamming distance of vectors $\bU$ and $\bV$.

A measurement matrix is ideal if it can be used to efficiently classify the given items in a noisy setting with a non-adaptive design, can be generated nonrandomly when there is little time and space to generate its entries, and has zero-gaps.

The numbers of choices for the types of defective items, the types of inhibitor items, the types of hybrid items, outcome setting, testing strategy, and criteria are $2^3 - 1, 2^3, 2^9, 2, 2, 2^6 - 1$, respectively. Therefore, the number of possible measurement matrices is $(2^3 - 1) \times 2^3 \times 2^9 \times 2 \times 2 \times (2^6 - 1) = 7,225,344$. Since the numbers of possible measurement matrices are extremely large, it would take time to design every measurement matrix. Therefore, a framework of encoding and decoding for the GGTI paradigm should be considered.

\subsection{Encoding procedure}
\label{sub:encGGTI}

To design a measurement matrix, we introduce notation for a perfect pair of two matrices.
\begin{definition}
Let $n, D, H, B$ be the number of items, the defective set, the inhibitor set, and the hybrid set, where $|D| \leq d$, $|H| \leq h$, $|B| \leq b$, and $1 \leq d + h + b \leq n$. Let $T = D \cup H \cup B$, $[n] \setminus T$ be the set of negative items, and $m_0 = \max \{h, d, b \}$. Let $\bV$ be an $n \times 1$ vector. Matrices $\cG$ and $\cM$ are considered to be a \textit{perfect pair} if the following conditions holds:
\begin{itemize}
\item There is an index set of rows of $\cG$, denoted $D^\prime$, such that $\supp(\cG_{i, *}) \cap T$ satisfies property $\Upsilon_i$ for $i \in D^\prime$, and $D \subseteq \bigcup_{i \in D^\prime} \left( \supp(\cG_{i, *}) \cap T \right)$. Set $\Upsilon = \{\Upsilon_1, \ldots, \Upsilon_{|D^\prime|} \}$.
\item There is an index set of rows of $\cG$, denoted $H^\prime$, such that $\supp(\cG_{i, *}) \cap T$ satisfies property $\Phi_i$ for $i \in H^\prime$, and $H \subseteq  \bigcup_{i \in H^\prime} \left( \supp(\cG_{i, *}) \cap T \right)$. Set $\Phi = \{\Phi_1, \ldots, \Phi_{|H^\prime|} \}$.
\item There is an index set of rows of $\cG$, denoted $B^\prime$, such that $\supp(\cG_{i, *})  \cap T$ satisfies property $\Psi_i$ for $i \in B^\prime$, and $B \subseteq \bigcup_{i \in B^\prime} \left( \supp(\cG_{i, *}) \cap T \right)$. Set $\Psi = \{\Psi_1, \ldots, \Psi_{|B^\prime|} \}$.
\item If $|\bV| \leq m_0$ and $\supp(\bV)$ satisfies some property in $\Upsilon$ (resp., $\Phi$, $\Psi$), vector $\bV$ can always be recovered from $\cM \bullet \bV$, and denoted $\bV = \dec_{\Upsilon}(\cM \bullet \bV, \cM)$ (resp., $\bV = \dec_{\Phi}(\cM \bullet \bV, \cM)$, $\bV = \dec_{\Psi}(\cM \bullet \bV, \cM)$).
\item If $|\bV| > m_0$ or $\supp(\bV)$ does not satisfy any property in $\Upsilon, \Phi$, or $\Psi$, then $\dec_{\Upsilon}(\cM \bullet \bV, \cM), \dec_{\Phi}(\cM \bullet \bV, \cM)$, or $\dec_{\Psi}(\cM \bullet \bV, \cM)$ either cannot return $\bV$ or returns vector $\bV^\prime \neq \bV$, where $|\bV^\prime| \leq m_0$.
\end{itemize}
\label{def:perfecPair}
\end{definition}

Note that in the above definition, matrix $\cM$ can ``handle'' all three properties, $\Upsilon, \Phi$, and $\Psi$. We thus create a measurement matrix $\cT = \cG \circledcirc \cM$, where $\circledcirc$ is defined in~\eqref{eqn:tensor}. Let $\bX = (x_1, \ldots, x_n)^T$ be the representation vector of $n$ items, where $x_j = 0$ means item $j$ is negative and item $j$ is either defective, inhibitory, or hybrid, otherwise. Note that $\supp(\bX) = T$. Then the outcome vector, $\bY = \cT \bullet \bX$, obtained by using $\bX$ and $\cT$ is

\begin{align}
\mathbf{y} = \begin{bmatrix}
\cM \times \diag(\cG_{1, *}) \\
\vdots \\
\cM \times \diag(\cG_{g, *})
\end{bmatrix} \bullet \bX 
= \begin{bmatrix}
(\cM \times \diag(\cG_{1, *}) \bullet \bX \\
\vdots \\
(\cM \times \diag(\cG_{g, *}) \bullet \bX
\end{bmatrix} = \begin{bmatrix}
\cM \bullet (\diag(\cG_{1, *}) \times \bX) \\
\vdots \\
\cM \bullet (\diag(\cG_{g, *}) \times \bX)
\end{bmatrix}
= \begin{bmatrix}
\bY_1 \\
\vdots \\
\bY_g
\end{bmatrix}, \label{eq:transform}
\end{align}
where $\bY_i = \cM \bullet (\diag(\cG_{i, *}) \times \bX)$ for $i = 1, \ldots, g$. We got~\eqref{eq:transform} because 
\begin{align}
&(\cM \times \diag(\cG_{i, *})) \bullet \bX && \nonumber \\
&= \begin{bmatrix}
\test( (\supp(\cM_{1, *}) \cap \supp(\cG_{i, *})) \cap \supp(\bX) ) \\
\vdots \\
\test( (\supp(\cM_{k, *}) \cap \supp(\cG_{i, *})) \cap \supp(\bX) )
\end{bmatrix} 
&=& \begin{bmatrix}
\test( \supp(\cM_{1, *}) \cap (\supp(\cG_{i, *}) \cap \supp(\bX))) \\
\vdots \\
\test( \supp(\cM_{k, *}) \cap (\supp(\cG_{i, *}) \cap \supp(\bX)))
\end{bmatrix} \nonumber \\
&= \begin{bmatrix}
\test( \supp(\cM_{1, *}) \cap \supp(\diag(\cG_{i, *}) \times \bX) ) \\
\vdots \\
\test( \supp(\cM_{k, *}) \cap \supp(\diag(\cG_{i, *}) \times \bX) )
\end{bmatrix}
&=& \cM \bullet (\diag(\cG_{i, *}) \times \bX). \nonumber
\end{align}

We make two examples for a perfect pair for property $\Upsilon_i, \Phi_i, \Psi_i$. First, in the CGT model, $H = B = \emptyset$. Therefore, we pay attention to only defective set $D$. In the scheme proposed by Cai et al.~\cite{cai2013grotesque}, $\Upsilon_i$ is ``The cardinality of $\supp(\cG_{i, *}) \cap T$ is one'' for every $i \in D^\prime$. Second, in the GTDI, $B = \emptyset$. In the scheme proposed by Bui et al.~\cite{bui2018sublinear}, $\Upsilon_i$ is ``The cardinality of $\supp(\cG_{i, *}) \cap T$ is one'' for every $i \in D^\prime$, and $\Phi_i$ is ``There are one known defective item and one unknown inhibitor in $\supp(\cG_{i, *}) \cap T$'' for every $i \in H^\prime$.

\subsection{Decoding procedure}
\label{sub:decGGTI}

Let a $g \times n$ matrix $\cG$ and a $k \times n$ matrix $\cM$ be a perfect pair as defined in Definition~\ref{def:perfecPair} and let $\cT = \cG \circledcirc \cM$.  We can recover $D, H, B$ with at most $g \times \max \{d, h, b \}$ misclassified items in each set by using Algorithm~\ref{alg:decodingThreshold}.

\begin{algorithm}
\caption{$\mathrm{dec}(\bY, \cM)$: Decoding algorithm for generalized group testing with inhibitors}
\label{alg:decodingThreshold}
\textbf{Input:} Outcome vector $\bY$, matrix $\cM$.\\
\textbf{Output:} Set of defectives, inhibitors, hybrid items, and possibly negative items.

\begin{algorithmic}[1]
\State $S_1 = S_2 = S_3 = \emptyset$. \label{alg:init}
\For {$i=1$ to $g$} \label{alg:scan}
	\State $S_1 = S_1 \cup \supp \left( \dec_{\Upsilon}(\bY_i, \cM) \right)$. \label{alg:decodeDefectives}
	\State $S_2 = S_2 \cup \supp \left( \dec_{\Psi}(\bY_i, \cM) \right)$. \label{alg:decodeInhibitors}
	\State $S_3 = S_3 \cup \supp \left( \dec_{\Phi}(\bY_i, \cM) \right)$. \label{alg:decodeHybrid}
\EndFor \label{alg:endScan}
\State Return $S$. \label{alg:defectiveSet}
\end{algorithmic}
\end{algorithm}

From Algorithm~\ref{alg:decodingThreshold}, we derive the following theorem

\begin{theorem}
\label{thr:frameworkGGTI}
Let $1 \leq d, h, b \leq n$ be integers and $d + h + b \leq n$. Let $n, D, H, B, [n] \setminus T = D \cup H \cup B$ be the number of items, the defective set, the inhibitor set, the hybrid set, and the set of negative items, where $|D| \leq d$, $|H| \leq h$, and $|B| \leq b$. Suppose that a $g \times n$ matrix $\cG$ and a $k \times n$ matrix $\cM$ are a perfect pair as defined in Definition~\ref{def:perfecPair} and that matrix $\cM$ can be decoded in time $O(\sfA)$. Then a measurement matrix $\cT = \cG \circledcirc \cM$ can be used to identify $D, H, B$ in time $O(g \sfA)$ with at most $g \times \max \{d, h, b \}$ misclassified items in each set recovered.
\end{theorem}

\begin{proof}
We have $\supp ( \diag(\cG_{i, *}) \times \bX) = \supp(\cG_{i, *}) \cap \supp(\bX) = \supp(\cG_{i, *}) \cap T$. Therefore, for any $i \in D^\prime$, $\diag(\cG_{i, *}) \times \bX$ satisfies property $\Upsilon_i$ because of the definition of $\cG$. Moreover, we have $|\supp ( \diag(\cG_{i, *}) \times \bX)| = |\supp(\cG_{i, *}) \cap T| \leq m_0$. Therefore, we get $\dec_{\Upsilon}(\bY_i, \cM) = \diag(\cG_{i, *}) \times \bX$ for all $i \in D^\prime$. From the definition of $\cG$, $D \subseteq \bigcup_{i \in D^\prime} \left( \supp(\cG_{i, *}) \cap T \right) = \bigcup_{i \in D^\prime} \supp( \dec(\bY_i, \cM) )$. Thus, performing Steps~\ref{alg:scan} to~\ref{alg:endScan} in Algorithm~\ref{alg:decodingThreshold} ensures that the set $S_1$ contains $D$. Similarly, it can be proved that $S_2$ and $S_3$ contain $H$ and $B$, respectively.

Because it takes time $O(\sfA)$ to perform Steps~\ref{alg:decodeDefectives} to~\ref{alg:decodeHybrid} and the number of loops in Step~\ref{alg:scan} is $g$, the running time of the algorithm is $O(g \sfA)$. To estimate the cardinality of $S_1, S_2$, and $S_3$, it is noted that $|\dec_{\Upsilon}(\bV, \cM)|$, $|\dec_{\Phi}(\bV, \cM)|$, or $|\dec_{\Psi}(\bV, \cM)|$ is not exceeded $m_0 = \max \{d, h, b \}$ for any vector $\bV$. We thus have $|S_1|, |S_2|, |S_3| \leq g m_0$. Therefore, the measurement matrix $\cT$ can be used to identify $D, H, B$ in time $O(g \sfA)$ with at most $g \times \max \{d, h, b \}$ misclassified items in each set recovered.
\end{proof}

The method presented in Theorem~\ref{thr:frameworkGGTI} is simply a possible method for creating a measurement matrix. Concrete methods for generating a measurement matrix need to be developed. Moreover, misclassified items can be removed for some instances of GGTI~\cite{bui2018efficient,bui2018efficiently,bui2019improved,bui2018sublinear}.

\section{Application to neuron classification}
\label{sec:neuro}

\subsection{Types of neurons}
\label{sub:neuro:types}

As described in section~\ref{sec:intro}, there are four types of neurons for a stimulus: excitatory, inhibitory, hybrid, and negative. A neuron generating an AP must be excitatory, inhibitory, or hybrid. Negative neurons do not generate APs. Therefore, we formally define four types of neurons:

\begin{enumerate}
\item A negative neuron does not generate an AP.
\item An excitatory neuron generates an AP and propagates that AP to another neurons.
\item An inhibitory neuron generates an AP and inhibits other neurons from generating APs.
\item A hybrid neuron generates an AP, propagates that AP to another neurons, and inhibits other neurons from generating APs.
\end{enumerate}

Although it is difficult to have a hybrid neuron for a stimulus, a hybrid neuron can be detected by observing different stimuli. For some stimuli, a hybrid neuron plays the role of an excitatory neuron. For other stimuli, a hybrid neuron plays the role of an inhibitor.

\subsection{Mapping between generalized group testing with inhibitors and classification of neurons}
\label{sub:neuro:mapping}

It is natural to do mapping between GGTI and neuron types. Excitatory neurons, inhibitory neurons, hybrid neurons, and negative neurons are represented by defective items, inhibitor items, hybrid items, and negative items in GGTI. Each type of neuron has corresponding sub-types of items in GGTI. However, because of the ``all or none'' characteristic of APs, the outcome of a test on a subset of neurons is either positive, i.e., an AP is generated corresponding to the synaptic integration of the neurons, or negative, otherwise. It is equivalent to the case that there is \textit{no gap} in GGTI for TGT, CmplxGT, GTTI, and GTCI. Because excitation and inhibition are inseparable events in the brain and synaptic integration always occurs, the model selected must have defective and inhibitor items. The model selected for a neural network is transformed into a model in GGTI, as illustrated in Figure~\ref{fig:GGT}. A corresponding measurement matrix is then generated. The measurement matrix is then used to perform tests and obtain test outcomes. Finally, the type of each neuron is identified after performing a decoding procedure.

The main advantage of using GGTI for classifying neurons is that the number of tests is usually much smaller than the number of neurons. It would save time and money when deploying the practical model.

\section{Design of tests in practice}
\label{sec:design}

\subsection{Ideal solution}
\label{sub:design:netMap}

We present an ideal solution to \textit{exactly} classify neurons. A test on a subset of neurons is deployed as follows. Let assume that all neurons in the test have synapses with \text{some neuron}, called a pooled neuron (steps 1-3 in Fig.~\ref{fig:GGT}). Then we gather data from the pooled neuron (steps 3-5 in Fig.~\ref{fig:GGT}). Because an AP is an all or none phenomenon, it is easy to convert (step 6 in Fig.~\ref{fig:GGT}) the gathered data to binary outcomes which is the step 7 in Fig.~\ref{fig:GGT}. Then the outcome of a test on those neurons in the theoretical model is equivalent to the outcome of a test on the pooled neuron.  Finally, all neurons are classified in steps 7 to 9 in Fig.~\ref{fig:GGT}.

\textbf{Discussion:} Some issues are arisen in this approach. First, the connections of the neurons in our experiments to other neurons must be constructed. Fortunately, Markram et al.~\cite{markram2015reconstruction} show the possibility that we can do that (though the number of connections and synapses is more than millions). Let assume that the connections of the neurons is achievable. The other issue is that there is no guarantee that there exists such pooled neurons in reality. Therefore, we should have a quantity to measure how close the graph induced by the neurons connections and the graph induced by the measurement matrix we used is. Moreover, the signals induced at pooled neurons may have some noise. Therefore, the measurement matrix should be error-tolerant.

\subsection{Feasible solutions}
\label{sub:design:feasible}

We present a feasible solution here. In this case, we can only classify negative neurons and non-negative neurons. Negative neurons do not generate APs while non-negative neurons generate APs. After classifying these neurons, it is possible to identify which neuron is excitatory, inhibitory, or hybrid by using the multi-electrode patch-clamp system~\cite{perin2013computer}.

Suppose that each neuron is connected and characterized by a conduction fiber. The goal of the conduction fibe is to conduct signal from the targeted neuron. In reality, we have an advanced devices to keep tracking simultaneously up to 10 neurons using multi-electrode patch-clamp system~\cite{perin2013computer}. However, a patch-clamp technique is too powerful to use in this case. We suggest to replace the multi-electrode patch-clamp system by a system of conduction fibers. It is economical and easy to implement.

A pooled neuron in section~\ref{sub:design:netMap} is now replaced by a multimeter. It is equivalent to the fact that the number of multimeters is exactly the number of tests in the theoretical model. Similar to section~\ref{sub:design:netMap}, the data gathered from the multimeters can be easily converted into binary outcomes, which are the outcomes of tests in the theoretical model, because an AP is an all or none phenomenon. Then negative and non-negative neurons can be classified.

\section{Conclusions}
\label{sec:cls}

Our motivation is to develop an efficient scheme for classifying neurons. We have first presented a generalization of group testing with inhibitors by introducing hybrid items. A framework for encoding and decoding is then presented to solve this generalization. We then apply the proposed framework for classifying neurons in theory. Finally, we show how to implement our proposed scheme in practice from an in vivo approach. It is interesting to determining the capacity for GGTI in terms of number of tests and decoding time. Moreover, it is an open problem that whether there exists a scheme such that there are no misclassified items in the recovered sets for GGTI.

\bibliographystyle{ieeetr}
\bibliography{bibli}

\end{document}